\documentclass[pra,aps,twocolumn,twoside,superscriptaddress]{revtex4}

\usepackage{amsmath,amsfonts,amssymb,caption,color,epsfig,graphics,graphicx,hyperref,latexsym,mathrsfs,revsymb,theorem,url,verbatim,enumerate,epstopdf,tikz,float,multirow,booktabs,appendix}
\hypersetup{colorlinks,linkcolor={blue},citecolor={red},urlcolor={blue}}
\usetikzlibrary{arrows, decorations.markings}

\tikzstyle{vecArrow} = [thick, decoration={markings,mark=at position
	1 with {\arrow[semithick]{open triangle 60}}},
double distance=1.4pt, shorten >= 5.5pt,
preaction = {decorate},
postaction = {draw,line width=1.4pt, white,shorten >= 4.5pt}]
\tikzstyle{innerWhite} = [semithick, white,line width=1.4pt, shorten >= 4.5pt]

\newtheorem{definition}{Definition}
\newtheorem{proposition}[definition]{Proposition}
\newtheorem{lemma}[definition]{Lemma}

\newtheorem{theorem}[definition]{Theorem}
\newtheorem{corollary}[definition]{Corollary}
\newtheorem{conjecture}[definition]{Conjecture}

\newtheorem{remark}[definition]{Remark}
\newtheorem{example}[definition]{Example}
\newtheorem{question}[definition]{Question}

\def\bcj{\begin{conjecture}}
	\def\ecj{\end{conjecture}}
\def\bcr{\begin{corollary}}
	\def\ecr{\end{corollary}}
\def\bd{\begin{definition}}
	\def\ed{\end{definition}}
\def\bea{\begin{eqnarray}}
\def\eea{\end{eqnarray}}
\def\bem{\begin{enumerate}}
	\def\eem{\end{enumerate}}
\def\bex{\begin{example}}
	\def\eex{\end{example}}
\def\bim{\begin{itemize}}
	\def\eim{\end{itemize}}
\def\bl{\begin{lemma}}
	\def\el{\end{lemma}}
\def\bma{\begin{bmatrix}}
	\def\ema{\end{bmatrix}}
\def\bpf{\begin{proof}}
	\def\epf{\end{proof}}
\def\bpp{\begin{proposition}}
	\def\epp{\end{proposition}}
\def\bqu{\begin{question}}
	\def\equ{\end{question}}
\def\br{\begin{remark}}
	\def\er{\end{remark}}
\def\bt{\begin{theorem}}
	\def\et{\end{theorem}}

\def\squareforqed{\hbox{\rlap{$\sqcap$}$\sqcup$}}
\def\qed{\ifmmode\squareforqed\else{\unskip\nobreak\hfil
		\penalty50\hskip1em\null\nobreak\hfil\squareforqed
		\parfillskip=0pt\finalhyphendemerits=0\endgraf}\fi}
\def\endenv{\ifmmode\;\else{\unskip\nobreak\hfil
		\penalty50\hskip1em\null\nobreak\hfil\;
		\parfillskip=0pt\finalhyphendemerits=0\endgraf}\fi}
\newenvironment{proof}{\noindent \textbf{{Proof.~} }}{\qed}
\def\Dbar{\leavevmode\lower.6ex\hbox to 0pt
	{\hskip-.23ex\accent"16\hss}D}
\makeatletter
\def\url@leostyle{%
	\@ifundefined{selectfont}{\def\UrlFont{\sf}}{\def\UrlFont{\small\ttfamily}}}
\makeatother
\def\bcj{\begin{conjecture}}
	\def\ecj{\end{conjecture}}
\def\bcr{\begin{corollary}}
	\def\ecr{\end{corollary}}
\def\bd{\begin{definition}}
	\def\ed{\end{definition}}
\def\bea{\begin{eqnarray}}
\def\eea{\end{eqnarray}}
\def\bem{\begin{enumerate}}
	\def\eem{\end{enumerate}}
\def\bex{\begin{example}}
	\def\eex{\end{example}}
\def\bim{\begin{itemize}}
	\def\eim{\end{itemize}}
\def\bl{\begin{lemma}}
	\def\el{\end{lemma}}
\def\bpf{\begin{proof}}
	\def\epf{\end{proof}}
\def\bpp{\begin{proposition}}
	\def\epp{\end{proposition}}
\def\bqu{\begin{question}}
	\def\equ{\end{question}}
\def\br{\begin{remark}}
	\def\er{\end{remark}}
\def\bt{\begin{theorem}}
	\def\et{\end{theorem}}

\def\btb{\begin{tabular}}
	\def\etb{\end{tabular}}

\newcommand{\nc}{\newcommand}

\def\d{\delta}

\nc{\bbA}{\mathbb{A}} \nc{\bbB}{\mathbb{B}} \nc{\bbC}{\mathbb{C}}
\nc{\bbD}{\mathbb{D}} \nc{\bbE}{\mathbb{E}} \nc{\bbF}{\mathbb{F}}
\nc{\bbG}{\mathbb{G}} \nc{\bbH}{\mathbb{H}} \nc{\bbI}{\mathbb{I}}
\nc{\bbJ}{\mathbb{J}} \nc{\bbK}{\mathbb{K}} \nc{\bbL}{\mathbb{L}}
\nc{\bbM}{\mathbb{M}} \nc{\bbN}{\mathbb{N}} \nc{\bbO}{\mathbb{O}}
\nc{\bbP}{\mathbb{P}} \nc{\bbQ}{\mathbb{Q}} \nc{\bbR}{\mathbb{R}}
\nc{\bbS}{\mathbb{S}} \nc{\bbT}{\mathbb{T}} \nc{\bbU}{\mathbb{U}}
\nc{\bbV}{\mathbb{V}} \nc{\bbW}{\mathbb{W}} \nc{\bbX}{\mathbb{X}}
\nc{\bbZ}{\mathbb{Z}}

\nc{\bA}{{\bf A}} \nc{\bB}{{\bf B}} \nc{\bC}{{\bf C}}
\nc{\bD}{{\bf D}} \nc{\bE}{{\bf E}} \nc{\bF}{{\bf F}}
\nc{\bG}{{\bf G}} \nc{\bH}{{\bf H}} \nc{\bI}{{\bf I}}
\nc{\bJ}{{\bf J}} \nc{\bK}{{\bf K}} \nc{\bL}{{\bf L}}
\nc{\bM}{{\bf M}} \nc{\bN}{{\bf N}} \nc{\bO}{{\bf O}}
\nc{\bP}{{\bf P}} \nc{\bQ}{{\bf Q}} \nc{\bR}{{\bf R}}
\nc{\bS}{{\bf S}} \nc{\bT}{{\bf T}} \nc{\bU}{{\bf U}}
\nc{\bV}{{\bf V}} \nc{\bW}{{\bf W}} \nc{\bX}{{\bf X}}
\nc{\bZ}{{\bf Z}} \nc{\bm}{{\bf m}} \nc{\bv}{{\bf v}}
\nc{\ba}{{\bf a}} \nc{\be}{{\bf e}} \nc{\bu}{{\bf u}}
\nc{\brr}{{\bf r}} \nc{\bc}{{\bf c}}

\nc{\cA}{{\cal A}} \nc{\cB}{{\cal B}} \nc{\cC}{{\cal C}}
\nc{\cD}{{\cal D}} \nc{\cE}{{\cal E}} \nc{\cF}{{\cal F}}
\nc{\cG}{{\cal G}} \nc{\cH}{{\cal H}} \nc{\cI}{{\cal I}}
\nc{\cJ}{{\cal J}} \nc{\cK}{{\cal K}} \nc{\cL}{{\cal L}}
\nc{\cM}{{\cal M}} \nc{\cN}{{\cal N}} \nc{\cO}{{\cal O}}
\nc{\cP}{{\cal P}} \nc{\cQ}{{\cal Q}} \nc{\cR}{{\cal R}}
\nc{\cS}{{\cal S}} \nc{\cT}{{\cal T}} \nc{\cU}{{\cal U}}
\nc{\cV}{{\cal V}} \nc{\cW}{{\cal W}} \nc{\cX}{{\cal X}}
\nc{\cZ}{{\cal Z}}

\nc{\hA}{{\hat{A}}} \nc{\hB}{{\hat{B}}} \nc{\hC}{{\hat{C}}}
\nc{\hD}{{\hat{D}}} \nc{\hE}{{\hat{E}}} \nc{\hF}{{\hat{F}}}
\nc{\hG}{{\hat{G}}} \nc{\hH}{{\hat{H}}} \nc{\hI}{{\hat{I}}}
\nc{\hJ}{{\hat{J}}} \nc{\hK}{{\hat{K}}} \nc{\hL}{{\hat{L}}}
\nc{\hM}{{\hat{M}}} \nc{\hN}{{\hat{N}}} \nc{\hO}{{\hat{O}}}
\nc{\hP}{{\hat{P}}} \nc{\hR}{{\hat{R}}} \nc{\hS}{{\hat{S}}}
\nc{\hT}{{\hat{T}}} \nc{\hU}{{\hat{U}}} \nc{\hV}{{\hat{V}}}
\nc{\hW}{{\hat{W}}} \nc{\hX}{{\hat{X}}} \nc{\hZ}{{\hat{Z}}}

\nc{\hn}{{\hat{n}}}

\def\dim{\mathop{\rm Dim}}

\def\ghz{\mathop{\rm GHZ}}

\def\min{\mathop{\rm min}}

\newcommand{\tr}{{\rm Tr}}

\def\oa{\mathop{\rm OA}}

\def\iro{\mathop{\rm IrOA}}

\def\ame{\mathop{\rm AME}}

\newcommand{\bra}[1]{\langle#1|}
\newcommand{\ket}[1]{|#1\rangle}

\newcommand{\ketbra}[2]{|#1\rangle\!\langle#2|}
\newcommand{\braket}[2]{\langle#1|#2\rangle}

\newcommand{\fl}[2]{\lfloor\frac{#1}{#2}\rfloor}

\def\Dbar{\leavevmode\lower.6ex\hbox to 0pt
	{\hskip-.23ex\accent"16\hss}D}

\begin{document}
	
	\title{$k$-Uniform states and quantum information masking}
	


\author{Fei Shi}
\email[]{shifei@mail.ustc.edu.cn}
\affiliation{School of Cyber Security,
	University of Science and Technology of China, Hefei, 230026, People's Republic of China.}

\author{Mao-Sheng Li}\email[]{li.maosheng.math@gmail.com}
\affiliation{Department of Physics, Southern University of Science and Technology, Shenzhen 518055, China}
\affiliation{Department of Physics, University of Science and Technology of China, Hefei 230026, China}

\author{Lin Chen}
\email[]{linchen@buaa.edu.cn}
\affiliation{School of Mathematical Sciences, Beihang University, Beijing 100191, China}
\affiliation{International Research Institute for Multidisciplinary Science, Beihang University, Beijing 100191, China}

\author{Xiande Zhang}
\email[]{Corresponding author: drzhangx@ustc.edu.cn}
\affiliation{School of Mathematical Sciences,
	University of Science and Technology of China, Hefei, 230026, People's Republic of China}

\begin{abstract}
A pure state of $N$ parties with local dimension $d$ is called a $k$-uniform
state if  all the reductions to $k$ parties are maximally mixed. Based on the connections among $k$-uniform states, orthogonal arrays and linear codes, we give general constructions for $k$-uniform states. We show that when $d\geq 4k-2$ (resp. $d\geq 2k-1$) is a prime power, there exists a $k$-uniform state for any $N\geq 2k$ (resp. $2k\leq N\leq d+1$). Specially, we give the existence of $4,5$-uniform states for almost every $N$-qudits. Further, we generalize the concept of quantum information masking in bipartite systems given by [Modi \emph{et al.} \href{https://journals.aps.org/prl/abstract/10.1103/PhysRevLett.120.230501}{Phys. Rev. Lett. \textbf{120}, 230501 (2018)}] to $k$-uniform quantum information masking in multipartite systems, and we show that $k$-uniform states and quantum error-correcting codes can be used for $k$-uniform quantum information masking.

\end{abstract}
\maketitle

\renewcommand\arraystretch{1}
\begin{table*}[htbp]
	\caption{Existence of $4$-uniform states of $N$ subsystems with  local dimension $d\geq 2$.}
	\centering\label{table:4uni}
	\renewcommand\tabcolsep{10.0pt}
	\begin{tabular}{c|ccccccccc}
		\midrule[1.1pt]	
		$d \diagdown N$  &8         &9          &10        &11          &12        &13        &14       &15        &$N\geq 16$ \\
		\hline
		2                &$\times$  &$\times$   &$\times$  &?           &$\surd$   &$\surd$   &$\surd$  &$\surd$   &$\surd$ \\
		3                &$\times$  &$\surd$    &$\surd$   &$\surd$     &$\surd$   &$\surd$   &$\surd$  &$\surd$   &$\surd$ \\
		4,12                &?         &$\surd$    &$\surd$   &$\surd$     &$\surd$   &$\surd$   &$\surd$  &$\surd$   &$\surd$ \\
		6,10                &?         &?          &?         &?           &$\surd$   &$\surd$   &$\surd$  &$\surd$   &$\surd$ \\
		$d\geq 5$ is a prime power &$\surd$ &$\surd$ &$\surd$ &$\surd$ &$\surd$ &$\surd$   &$\surd$  &$\surd$   &$\surd$ \\
		$d\geq 14$ is not a prime power  &?   &?    &?   &?     &$\surd$   &$\surd$   &$\surd$  &$\surd$   &$\surd$ \\
		\midrule[1.1pt]
	\end{tabular}
\end{table*}

\begin{table*}[htbp]
	\caption{Existence of $5$-uniform states of $N$ subsystems with  local dimension $d\geq 2$.}
	\centering\label{table:5uni}
	\renewcommand\tabcolsep{10.0pt}
	\begin{tabular}{c|ccccccccc}
		\midrule[1.1pt]	
		$d \diagdown N$     &10        &11          &12        &13        &14       &15  &16         &17       &$N\geq 18$ \\
		\hline
		2                &$\times$  &$\times$   &?  &?        &?  &?   &$\surd$  &?  &$\surd$ \\
		3,15                &$\surd$  &?     &$\surd$  &?     &$\surd$   &$\surd$   &$\surd$  &$\surd$   &$\surd$ \\
		4,12                &$\surd$  &?     &$\surd$  &?     &$\surd$   &?   &$\surd$  &$\surd$   &$\surd$ \\
		5               &$\surd$   &?    &$\surd$   &$\surd$     &$\surd$   &$\surd$   &$\surd$  &$\surd$   &$\surd$ \\
		6,10,14                &?         &?          &?         &?           &?   &?   &$\surd$  &?   &$\surd$ \\
		$d\geq 7$ is a prime power &$\surd$ &$\surd$ &$\surd$ &$\surd$ &$\surd$ &$\surd$   &$\surd$  &$\surd$   &$\surd$ \\
		$d\geq 18$ is not a prime power   &?       &?       &?       &?    &?         &?   &$\surd$  &?   &$\surd$  \\
		\midrule[1.1pt]
	\end{tabular}
\end{table*}

\section{Introduction}
\label{sec:int}

Multipartite entanglement has many applications in quantum information such as quantum teleportation \cite{bennett1993teleporting,bouwmeester1997experimental}, quantum key distribution \cite{ekert1991quantum,gisin2002quantum,bennett1992quantum}, superdense coding \cite{bennett1992communication} and quantum error correcting codes \cite{scott2004multipartite}. However, it is difficult to  quantify the level of entanglement in an arbitrary multipartite system. A very striking class of pure states called \emph{absolutely maximally entangled} (AME) states  has attracted much attention in recent years. These states are maximally entangled on every bipartition. AME states can be used to design holographic
quantum codes and perfect tensors \cite{pastawski2015holographic}. AME states can also be used  for threshold quantum secret sharing schemes, for parallel and open-destination teleportation protocols \cite{helwig2012absolute,helwig2013absolutely}. However, AME states are very rare for given local dimensions \cite{AMEtable}. Taking qubits as an example, AME states exist only for  $2$-, $3$-, $5$-, and $6$-qubits
\cite{scott2004multipartite,rains1999quantum,huber2017absolutely}.

A much more general concept is called $k$-uniform
states, which is defined to be a multipartite pure
state of $N$-qudits whose all reductions to $k$ parties are maximally mixed \cite{scott2004multipartite}.
AME states are special kinds of $k$-uniform states, more exactly,  the $\fl{N}{2}$-uniform states. Goyeneche \emph{et al.} first associated the $k$-uniform states with orthogonal arrays \cite{goyeneche2014genuinely}.  Moreover,  $k$-uniform states can also be  constructed from  Latin squares, symmetric matrices, graph states, quantum error correcting codes and classical error correcting codes   \cite{goyeneche2014genuinely,feng2017multipartite,goyeneche2015absolutely,scott2004multipartite,goyeneche2018entanglement,helwig2013absolutelygraph}. Very recently, the authors in Ref. \cite{raissi2019constructing} derived a new method to construct $k$-uniform states.  Although there are several methods to construct $k$-uniform states, a complete solution to the existence of  $k$-uniform states   is far from reach. It has been known that a $1$-uniform state exists for any $d\geq 2$ and $N\geq 2$, by the existence of $\ghz$ states.  Pang \emph{et al.} showed the existence of $2,3$-uniform states for almost every $N$-qudits \cite{pang2019two}.
However, there are a few results on the existence of  $k$-uniform states of $N$-qudits when $k\geq 4$. In this paper, we shall give general constructions of $k$-uniform states for $k\geq 4$ by using the linear codes.

Recently, Modi \emph{et al} proposed the concept of quantum information masking \cite{modi2018masking}.  This is a physical process that encodes quantum information into a bipartite system, while the information is completely unknown to each local system. They obtained a no-masking theorem:  an arbitrary quantum state cannot be masked. In \cite{li2019k-uniform}, the authors showed that quantum information masking in multipartite systems is possible.  In their masking protocol, it also required that the original information is inaccessible to each local system. However, collusion between
some subsystems would then reveal the encoded quantum
information \cite{modi2018masking}. Thus, a stronger version of quantum information masking is desirable. In this paper, we propose the concept of $k$-uniform quantum information masking in multipartite systems. It requires that the original information is  inaccessible to each $k$ subsystems. Specially, we refer to the $\fl{N}{2}$-uniform quantum information masking as the strong quantum information masking. We show that  if there exists a $(k+1)$-uniform state of $(N+1)$-qudits, then all states of one-qudit can be $k$-uniformly masked in an $N$-qudits system. We also show that when $N$ is odd, all states of one-qudit can be strongly masked in an $N$-qudits system provided that an AME state of $(N+1)$-qudits exists. However, when $N$ is even, we show that the strong quantum information masking is impossible, as a generalized no-masking theorem.  In the $k$-uniform quantum information masking scheme, if the reduction states of $k$ parties are proportional to identity, we show that a pure $((N,d,k+1))_d$ quantum error-correcting code is equivalent to that all states of one-qudit can be $k$-uniformly masked in an $N$-qudits system.

The rest of this paper is organized as follows. In Sec.~\ref{sec:k-uniform}, we introduce the connections among $k$-uniform states, orthogonal arrays and linear codes. We also review the existence of $1,2,3$-uniform states. In Sec.~\ref{sec:construction}, by using linear codes, we show that when $d\geq 4k-2$ (resp. $d\geq 2k-1$) is a prime power, there exists a $k$-uniform state in $(\bbC^d)^{\otimes N}$ for any $N\geq 2k$ (resp. $2k\leq N\leq d+1$).  Specially, we give the existence of $4,5$-uniform states for almost every $N$-qudits (see Table~\ref{table:4uni} and Table~\ref{table:5uni}),  by using
the technique derived in Sec.~\ref{sec:construction} and some previously known results.
In Sec.~\ref{sec:k-uniformmasking}, we propose the concept of $k$-uniform quantum information masking in multipartite systems, and we show that $k$-uniform states and quantum error-correcting codes can be used for $k$-uniform quantum information masking. Finally, we conclude in Sec.~\ref{sec:conclusion}.


\section{$k$-uniform states, orthogonal arrays and linear codes}\label{sec:k-uniform}
In this section, we introduce the preliminary knowledge and facts. First, we introduce the concept of $k$-uniform states \cite{scott2004multipartite}.

\begin{definition}
	A $k$-uniform state $\ket{\psi}$ in $\otimes_{\ell=1}^{N}\cH_{A_\ell}$ with $\dim\cH_{A_\ell}=d$ ($\ell=1,\ldots,N$), has the property that all reductions to $k$ parties are maximally mixed. That is,	for any subset $\cB=\{A_{j_1},A_{j_2},\ldots,A_{j_k}\}\subset \{A_1,A_2,\ldots,A_N\}$,
	\begin{equation}
	\rho_{\cB}=\tr_{\cB^c}\ketbra{\psi}{\psi}=\frac{1}{d^k}I_{\cB},
	\end{equation}
	where   $\cB^c$ denotes the set
	$$\{A_1,A_2,\ldots,A_N\}\setminus \cB,$$  $\tr_{\cB^c}$ is the partial trace operation, and $I_{\cB}$ denotes the identity operation acting on the Hilbert space $\otimes_{\ell=1}^{k}\cH_{A_{j_\ell}}$.
\end{definition}

\begin{table*}[htbp]
	\renewcommand\arraystretch{1.7}	
	\caption{Existence of $1,2,3$-uniform states of $N$ subsystems with  local dimension $d\geq 2$.}\label{Table:k-unfiormresults}
	\centering
	\renewcommand\tabcolsep{5pt}
	\begin{tabular}{ccccc}
		\midrule[1.1pt]
		$(\bbC^d)^{\otimes N}$  &Existence  &Nonexistence &Unknown   &References  \\
		\hline
		$1$-uniform  &$d\geq 2, N\geq 2$ &no  & no &\cite{goyeneche2014genuinely}  \\
		\hline
		$2$-uniform  &$d\geq 2,N\geq 4$ except $d=2,6$, $N=4$ &$d=2, N=4$ & $d=6$, $N=4$ &\cite{goyeneche2014genuinely,scott2004multipartite,li2019k-uniform,pang2019two,rains1999nonbinary,highu}  \\
		\hline
		\multirow{2}*{$3$-uniform} &\multirow{2}*{$d\geq 2$, $N\geq 6$,}  &\multirow{2}*{$d=2, N=7$} &\multirow{2}*{$d\geq 6$,  $d=2 \pmod 4$,}  &\multirow{2}*{\cite{li2019k-uniform,rains1999nonbinary,huber2017absolutely,helwig2013absolutely,raissi2019constructing,grassl2015quantum,pang2019two}}\\
		&except $d=2 \pmod 4$, $N=7$& &$N=7$ &\\
		\midrule[1.1pt]	
	\end{tabular}
\end{table*}

We denote $\otimes_{\ell=1}^{N}\cH_{A_\ell}$ as $(\bbC^d)^{\otimes N}$ for simplicity.  Due to the Schmidt decomposition of bipartite pure state, $k$ satisfies $k\leq \fl{N}{2}$.  If $k=\fl{N}{2}$, then $\ket{\psi}$ is also called an absolutely maximally entangled ($\ame$) state. A $k$-uniform state in $(\bbC^d)^{\otimes N}$ corresponds to a pure $((N,1,k+1))_d$ quantum error-correcting code (QECC)  \cite{scott2004multipartite}. See Appendix~\ref{appendix:QECC} for the definition of QECCs. For a pure $((N,1,k+1))_{d_1}$ code $\mathcal{C}_1$ and a pure $((N,1,k+1))_{d_2}$ code $\mathcal{C}_2$, the tensor product $\mathcal{C}_1\otimes \mathcal{C}_2$ is a pure $((N,1,k+1))_{d_1d_2}$ code \cite{rains1999nonbinary}. By the relation between $k$-uniform states and QECCs, we have the following lemma.

\begin{lemma}\label{lem:rec}
Let $\ket{\psi}_{A_1,A_2,\ldots, A_N}$  and $\ket{\phi}_{B_1,B_2,\ldots, B_N}$ be $k$-uniform states in $(\bbC^{d_1})^{\otimes N}$  and   $(\bbC^{d_2})^{\otimes N}$ respectively. Then the tensor product of them, i.e.,   $$\ket{\varphi}_{(A_1B_1),(A_2B_2),\ldots,(A_NB_N)}=\ket{\psi}_{A_1,A_2,\ldots ,A_N}\otimes\ket{\phi}_{B_1,B_2,\ldots B_N}$$ is a $k$-uniform state in $(\bbC^{d_1d_2})^{\otimes N}$.
\end{lemma}

 Lemma~\ref{lem:rec} provides a simple but very  useful construction. If we can construct  a $k$-uniform state in $(\bbC^{d})^{\otimes N}$ for any prime $d$, then we can construct a $k$-uniform state in $(\bbC^{d})^{\otimes N}$ for any local dimension $d$. Since $\ghz$ states are  $1$-uniform states, it means that a $1$-uniform state in $(\bbC^d)^{\otimes N}$ exists for any $d\geq 2$ and $N\geq 2$.   In \cite{pang2019two}, there are several unsolved cases for the existence of $2$-uniform states and $3$-uniform states.  The unsolved cases for $2$-uniform states in $(\bbC^d)^{\otimes N}$ are  $d=2,6$ and $N=4$; $d=10$ and $N=5$.  By \cite{highu}, we know that $2$-uniform states in $(\bbC^2) ^{\otimes 4}$ do not exist.  By using Lemma~\ref{lem:rec}, we can construct a $2$-uniform state in  $(\bbC^{10})^{\otimes 5}$ from a  $2$-uniform state in  $(\bbC^{2})^{\otimes 5}$  and a  $2$-uniform state in  $(\bbC^{5})^{\otimes 5}$. Thus the only unsolved case for $2$-uniform states in $(\bbC^d)^{\otimes N}$ is just $d=6$ and $N=4$.  The unsolved cases for $3$-uniform states in $(\bbC^d)^{\otimes N}$ are that $d\geq 6$ is not a prime power and $N=9,10,11$; $d=2,3$, $d\geq 6$ is not a prime power and $N=6$; $d=2,3,4,5$, $ d\geq 6$ is not a prime power and $N=7$.  By using Lemma~\ref{lem:rec}, we can construct a  $3$-uniform state in $(\bbC^{d})^{\otimes N}$  for any $d\geq 6$ not a prime power and $N=9,10,11$. A $3$-uniform state in $(\bbC^{d})^{\otimes 6}$ exists for any $d\geq 2$ by \cite{rains1999nonbinary}. Since $3$-uniform states in  $(\bbC^{2})^{\otimes 7}$ do not exist \cite{huber2017absolutely}, and there exists a $3$-uniform state in  $(\bbC^{d})^{\otimes 7}$ for any $d=3,5$ by \cite{AMEtable}, $d=4$ by \cite{raissi2019constructing}, we know that  a $3$-uniform state  in $(\bbC^{d})^{\otimes 7}$ exists for any $d\geq 2$ except for $d=2(2k+1)$ and $k\geq 1$ by Lemma~\ref{lem:rec}.  Thus the left unsolved cases for $3$-uniform states  in  $(\bbC^{d})^{\otimes N}$ are $d\geq 6, d=2 \pmod 4$ and $N=7$. See Table~\ref{Table:k-unfiormresults} for a summary of $1,2,3$-uniform states.

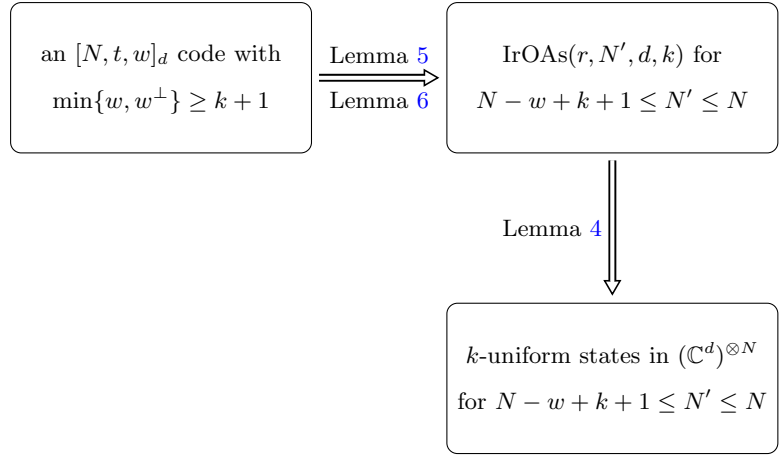
\begin{figure*}[t]
	\begin{tikzpicture}
	\draw[rounded corners] (0,0) rectangle (4,2);
	\draw (2,1.3) node []{an $[N,t,w]_d$ code with };
	\draw (2,0.7) node []{$\min\{w,w^{\bot}\}\geq k+1$};
	
	\draw[rounded corners] (5.8,0) rectangle (10.2,2);
	\draw (8,1.3) node []{IrOAs$(r,N',d,k)$ for};
	\draw (8,0.7) node []{ $N-w+k+1\leq N'\leq N$};
	
	\draw[rounded corners] (5.8,-4) rectangle (10.2,-2);
	\draw (8,-2.7) node []{$k$-uniform states in $(\bbC^d)^{\otimes N}$};
	\draw (8,-3.3) node []{for $N-w+k+1\leq N'\leq N$};

	\draw [vecArrow](4.1,1)--(5.7,1);
	\draw (4.9,1.3) node []{Lemma \ref{dis}};
	\draw (4.9,0.7) node []{Lemma \ref{lem:codearray}};
	\draw [vecArrow](8,-0.1)--(8,-1.9);
	\draw (7.2,-1) node []{Lemma~\ref{oakui}};
	\end{tikzpicture}
	\caption{The main method  of constructing $k$-uniform states in this paper. }\label{Fig:codeiroa}
\end{figure*}

Orthogonal arrays are essential in statistics and have wide applications in
computer science and cryptography.  An $r\times N$ array $A$ with entries taken from a set $S$ with $d$ elements is said to be an \emph{orthogonal array } \cite{hedayat1999orthogonal} with $r$ runs, $N$ factors, $d$ levels, strength $k$, and index $\lambda$, denoted by $\oa(r,N,d,k)$, if every $r\times k$ subarray of $A$ contains each $k$-tuple of symbols from $S$ exactly $\lambda$ times as a row.
It is easy to see that an $\oa(r,N,d,k)$ must be an $\oa(r,N,d,k-1)$, and it is an $\oa(r,N-n,d,k)$ when we delete any $n$ columns.

\begin{example}\label{OA9432}
	\begin{equation*}
	\left(
	\begin{matrix}
	0 &0 &0	&1 &1 &1 &2 &2 &2\\
	0 &1 &2	&0 &1 &2 &0 &1 &2\\
	0 &1 &2	&2 &0 &1 &1 &2 &0\\
	0 &1 &2 &1 &2 &0 &2 &0 &1\\
	\end{matrix}
	\right)^{\mathrm{T}}
	\end{equation*}
	is an $\oa(9,4,3,2)$.
\end{example}

An orthogonal array $\oa(r,N,d,k)$ is called \emph{irredundant} \cite{goyeneche2014genuinely}, denoted  by $\iro(r,N,d,k)$, if after removing from the array any $k$ columns all remaining $r$ rows, containing $N-k$ symbols each, are different. It is easy to see that the $\oa(9,4,3,2)$ in Example~\ref{OA9432} is an $\iro(9,4,3,2)$.

\begin{proposition}\label{oakui}\cite{goyeneche2014genuinely}
For any $\iro(r,N,d,k)$ say  $ (a_{ij})_{r\times N}$  with $a_{ij}\in \bbZ_d$,  the $N$ particles   state  $|\psi\rangle$    defined by  $|\psi\rangle:=\frac{1}{\sqrt{r}}\sum_{i=1}^{r}|a_{i1}a_{i2}\cdots a_{iN}\rangle$ is a $k$-uniform state in $(\bbC^{d})^{\otimes N}$.
\end{proposition}

By Example~\ref{OA9432} and Proposition \ref{oakui}, we can construct a $2$-uniform state in $(\bbC^{3})^{\otimes 4}$. In \cite{pang2019two}, the authors provided a way to check whether an orthogonal array is irredundant. For any two vectors $\bu$ and $\bv$ of length $N$, the \emph{Hamming distance} $d_H(\bu,\bv)$ is the  number of positions in which they differ. Given an $r\times N$ array $A$, the minimum distance of $A$, denoted by $d_H(A)$ is the minimum Hamming distance between any two distinct rows.

\begin{lemma}\label{dis}\cite{pang2019two}
	An $\oa(r,N,d,k)$ is irredundant if and only if its minimum distance is greater than $k$. The existence of an $\oa(r,N,d,k)$ with minimum distance $w\geq k+1$ implies that an $\iro(r,N',d,k)$ exists for any $N-w+k+1\leq N'\leq N$.
\end{lemma}

The set of rows of an $r\times N$ array $A$ over $\bbZ_d$ is often referred to a code $\mathcal{C}$ in $\bbZ_d^N$, where $\bbZ_d^N$ means $\overbrace{\bbZ_d \times\cdots\times\bbZ_d}^{N \text{ times}}$. If $d_H(A)=w$, then we say that the code $\mathcal{C}$ has minimum distance $w$, and denote it by an $(N,r,w)_d$ code. Further, if $d$ is a prime power, and $\mathcal{C}$ forms a subspace of $\bbF_d^N$ of dimension $t$ ($\bbF_d$ is a Galois Field with order $d$), then it is called a linear code, and denoted by an $[N,t,w]_d$ code. In this case, the dual code of $\mathcal{C}$ is defined by $\mathcal{C}^{\bot}=\{\bu\in \bbF_d^N | \bu^{\mathrm{T}}\cdot\bv=0, \forall \ \bv\in \mathcal{C}\}$, and the minimum distance of $\mathcal{C}^{\bot}$ is called the dual distance of $\mathcal{C}$, denoted by $w^{\bot}$. If $\mathcal{C}=\mathcal{C}^{\bot}$, then $\mathcal{C}$ is called \emph{self-dual}.  The next lemma gives the relation between linear codes and orthogonal arrays.

\begin{lemma}\cite{hedayat1999orthogonal}\label{lem:codearray}
	If $\mathcal{C}$ is an $[N,t,w]_d$ code with dual distance $w^{\bot}$, then the corresponding array is an $\oa(d^t,N,d,w^{\bot}-1)$.
\end{lemma}

Example~\ref{OA9432} is a $[4,2,3]_3$ code with dual distance $3$, then we  obtain an $\oa(9,4,3,2)$ by Lemma~\ref{lem:codearray}. If an orthogonal array is constructed from a linear code, we call it  a \emph{linear} orthogonal array.   The strategy for constructing $k$-uniform states in this paper is as follows (see Fig.~\ref{Fig:codeiroa}). Using some known linear codes and the correspondence in Lemma~\ref{lem:codearray}, we can obtain an  orthogonal array.
Then by  computing the minimum distance of the   orthogonal array, we might obtain a series of irredundant orthogonal arrays by Lemma~\ref{dis}.
Finally, by Proposition~\ref{oakui}, we obtain a series of $k$-uniform states.
From Lemma~\ref{lem:codearray}, we can give a recursion for $\iro$s by the recursion for linear codes.

\section{Construction of $k$-uniform states }\label{sec:construction}
In this section, we give a recursion for $\iro$s in Lemma~\ref{lem;linirr}. Then we construct some $k$-uniform states by linear codes in Theorem~\ref{thm:dpri}. Further, we give the existence of $4,5$-uniform states in Tables~\ref{table:4uni} and \ref{table:5uni}. First, we give the key recursion for $\iro$s. For any two positive integers $a$ and $b$ ($a\leq b$), let $[a,b]$ denote the set $\{a,a+1,\ldots,b\}$.

\begin{lemma}\label{lem;linirr}
	If there exists a linear $\iro(r_1,N_1,d,k)$ and a linear $\iro(r_2,N_2,d,k)$, then there exists a linear $\iro(r_1r_2,N_1+N_2,d,k)$.
\end{lemma}
\begin{proof}
		We claim that if  $\mathcal{C}_1$ is an $[N_1,t_1,w_1]_d$ code with dual distance $w_1^{\bot}$  and $\mathcal{C}_2$ is an $[N_2,t_2,w_2]_d$ code with dual distance $w_2^{\bot}$, then $\mathcal{C}_1\oplus \mathcal{C}_2=\{(\bc_1,\bc_2)|\bc_1\in\mathcal{C}_1,\bc_2\in\mathcal{C}_2\}$ is an $[N_1+N_2,t_1+t_2,\min\{w_1,w_2\}]_d$ code with dual distance  $\min\{w_1^{\bot},w_2^{\bot}\}$. See Appendix~\ref{appendix:prooflinear} for the proof of this claim. By Lemmas~\ref{dis} and \ref{lem:codearray}, we obtain that the linear $\iro(r_j,N_j,d,k)$  corresponds to an $[N_i,\text{log}_dr_j,w_j]_d$ code $\mathcal{C}_j$  with $\min\{w_j,w_j^{\bot}\}\geq k+1$ for each $j=1,2$. Then  $\mathcal{C}_1\oplus \mathcal{C}_2$ is an $[N_1+N_2,\text{log}_d(r_1r_2),\min\{w_1,w_2\}]_d$ code with dual distance  $\min\{w_1^{\bot},w_2^{\bot}\}$ by the above claim. Since the dual distance of this code is $\min\{w_1^{\bot},w_2^{\bot}\}\geq k+1$, this code corresponds to an $\oa(r_1r_2,N_1+N_2,d,k)$ with minimum distance $\min\{w_1,w_2\}\geq k+1$ by Lemma~\ref{lem:codearray}. Thus, it is an $\iro(r_1r_2,N_1+N_2,d,k)$ by Lemma~\ref{dis}.
\end{proof}
\vspace{0.4cm}
	
 By repeatedly using the above lemma, we can deduce that if there exists a linear  $\iro(r_1,N,d,k)$ for any $N\in[N_1,2N_1-1]$, then there exists a linear  $\iro(r_1,N,d,k)$  for any integer $N\geq N_1$. By some known linear codes, we construct some $k$-uniform states in the following.

\begin{theorem}\label{thm:dpri}
	If $d\geq 2k-1$ is a prime power, then there exists a $k$-uniform state in $(\bbC^{d})^{\otimes N}$ for any $N\in [2k, d+1]$; if $d\geq 4k-2$ is a prime power, then there exists a $k$-uniform state in $(\bbC^{d})^{\otimes N}$ for any  $N\geq 2k$.
\end{theorem}
\begin{proof}
	When $d$ is a prime power, there exists a $[d+1,k,d-k+2]_d$  code for any $1\leq k\leq d+1$ \cite[Theorem 9, Chapter 11]{macwilliams1977theory}, which is in fact a maximum distance separable (MDS) code \footnote{An $[N,t,w]_d$ code is called an MDS code if $w=N-t+1$.}. Since the dual code of an MDS code is still an MDS code,  the dual distance is $k+1$. By Lemma \ref{lem:codearray}, there exist a linear $\iro(r,d+1,d,k)$ with minimum distance $d-k+2$. Then there exists a linear $\iro(r,N',d,k)$ any $N\in [2k, d+1]$ when $d\geq 2k-1$ by Lemma~\ref{dis}. If $d\geq 4k-2$ is a prime power, then there exists a linear $\iro(r,N,d,k)$ for  any $N\geq 2k$ by Lemma~\ref{lem;linirr}.
\end{proof}
\vspace{0.4cm}	

 In \cite{feng2017multipartite}, the authors showed that there exists a constant $M_k(d)\geq 2k$, such that for all $N\geq M_k(d)$, there exists a $k$-uniform state in  $(\bbC^d)^{\otimes N}$  by probabilistic method. From Theorem~\ref{thm:dpri}, we know that $M_k(d)$ can be chosen to be $2k$ when $d\geq 4k-2$ is a prime power. Further, $M_k(d_1d_2\dots d_n)$ can be chosen to be $2k$ when every $d_i\geq 4k-2$ is a prime power by Lemma~\ref{lem:rec}.

\begin{figure*}[t]
	\begin{tikzpicture}
	\draw[rounded corners] (0,0) rectangle (2,4);
	\draw (-1.45,0.2)--(0,0.2) (-1.45,0.9)--(0,0.9) (-1.45,1.6)--(0,1.6)
	(-1.45,2.3)--(0,2.3) (-1.45,3)--(0,3)   (-1.45,3.7)--(0,3.7) ;
	\draw (-2.4,3.7) node []{$A_1$};
	\draw (-2.4,3) node []{$A_2$};
	\draw (-2.4,2.3) node []{$A_{\ell_1}$};
	\draw (-2.4,1.6) node []{$A_{\ell_2}$};
	\draw (-2.4,0.9) node []{$A_{\ell_k}$};
	\draw (-2.4,0.2) node []{$A_N$};
	\draw (-1.8,3.7) node []{$\ket{a_j}$};
	\draw (-1.8,3) node []{$\ket{b_2}$};
	\draw (-1.8,2.3) node []{$\ket{b_{\ell_1}}$};
	\draw (-1.8,1.6) node []{$\ket{b_{\ell_2}}$};
	\draw (-1.8,0.9) node []{$\ket{b_{\ell_k}}$};
	\draw (-1.8,0.2) node []{$\ket{b_N}$};
	\draw (-2.7,2.3)--(-4.1,1.6) (-2.7,1.6)--(-4.1,1.6) (-2.7,0.9)--(-4.1,1.6);
	\draw (-4.3,1.6) node []{$A$};	
	\draw (1,2) node []{$U_{\cS}$};
	\draw (-0.75,2.75)	node []{$\vdots$};
	\draw (-0.75,1.35)	node []{$\vdots$};
	\draw (-0.75,2.06)	node []{$\vdots$};
	\draw (-0.75,0.65)	node []{$\vdots$};
	\draw (2,0.2)--(3.5,0.2) (2,0.9)--(3.5,0.9) (2,1.6)--(3.5,1.6)
	(2,2.3)--(3.5,2.3) (2,3)--(3.5,3)   (2,3.7)--(3.5,3.7);
	\draw (3.5,2.3)--(5,1.6) (3.5,1.6)--(5,1.6) (3.5,0.9)--(5,1.6);
	\draw (5.3,1.6)	node []{$\rho_A$};
	\draw (2.75,2.75)	node []{$\vdots$};
	\draw (2.75,1.35)	node []{$\vdots$};
	\draw (2.75,2.06)	node []{$\vdots$};
	\draw (2.75,0.65)	node []{$\vdots$};
	\draw (-2.5,1.4)	node []{$\vdots$};

	\end{tikzpicture}
	\caption{Sketch map of $k$-uniform quantum information masking. The  state $\ket{a_j}$ is  to be encoded in the $k$-uniform quantum information masking process while $|b_i\rangle$ $(2\leq i\leq N$) are the initial states of the  ancillary systems. The figure shows that the reduction state of any fixed $k$ subsystems after the above process is independent of the encoded state $\ket{a_j}$. i.e. The reduction state is only subsystems dependent and we denote it to be    $\rho_A$ for subsystems $A=\{A_{l_1},\cdots,A_{l_k}\}$. }\label{Fig:defini9}
\end{figure*}
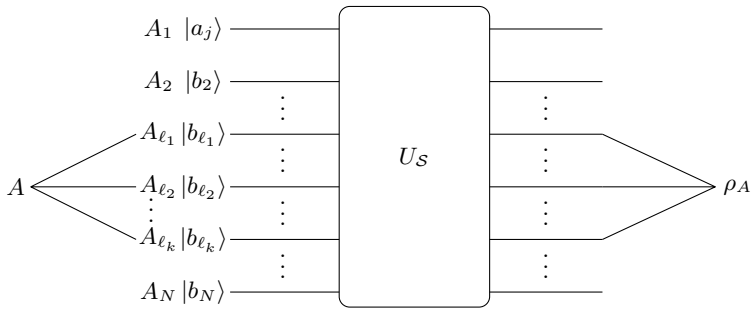

Next, we focus on the construction of $4,5$-uniform states. First, we consider $4$-uniform states.  Some of our constructions are from self-dual codes. For an $[N,\frac{N}{2},w]_s$ self-dual code, the dual distance is also $w$. For example, there exists a $[12,6,6]_4$ self-dual code \cite{Selftable}, then there exists a linear $\iro(r,N,2,4)$ for any $N\in [11,12]$ by Lemmas~\ref{dis} and \ref{lem:codearray}. See  \cite{Selftable} for  tables  of self-dual codes. By Theorem~\ref{thm:dpri}, when $d\geq 14$ is a prime power, there exists a $4$-uniform state in $(\bbC^{d})^{\otimes N}$ for any $N\geq 8$. By \cite[Theorem 12]{feng2017multipartite},  there exists a $4$-uniform state in $(\bbC^d)^{\otimes N}$ for any prime $d\geq 2$ and $N\geq 12$. Then there exists a $4$-uniform state in $(\bbC^d)^{\otimes N}$ for any $d\geq 2$ and $N\geq 12$ by Lemma~\ref{lem:rec}.  Thus, we only need to consider $d<14$ that is a prime power, and $8\leq N\leq 11$.
\begin{enumerate}[{(i)}]
\item When $d=2$, $4$-uniform states in $(\bbC^2)^{\otimes N}$ do not exist for each $N=8,9,10$  by Rains' bound \cite{rains1999quantum}.

\item When $d=3$,  $4$-uniform states in $(\bbC^3)^{\otimes 8}$ do not exist by Shadow bound \cite{huber2018bounds},  and a $4$-uniform state in $(\bbC^3)^{\otimes N}$ exists for any $N=9$ by \cite{AMEtable}, and $N=10,11$ by \cite{feng2017multipartite}.

\item When $d=4$, a $4$-uniform state in $(\bbC^4)^{\otimes N}$ exists for any  $N=9,10$ by \cite{AMEtable} and $N=11$ by a $[12,6,6]_4$ self-dual code in \cite{Selftable}.

\item When $d=5$, a $4$-uniform state in $(\bbC^5)^{\otimes N}$ exists for any $N\in[8, 10]$ by \cite{AMEtable},  $N=11$ by a $[12,6,6]_5$ self-dual code  in \cite{Selftable}.

\item When $d=7,8$,   a $4$-uniform state in $(\bbC^7)^{\otimes N}$ exists for any $N\in [8, 11]$ by \cite{AMEtable}.

\item When $d=9$, a $4$-uniform state in $(\bbC^9)^{\otimes N}$ exists for any $N\in [9,11]$ by Lemma~\ref{lem:rec}, and $N=8$ by Theorem~\ref{dis}.

\item When $d=11,13$, a $4$-uniform state in $(\bbC^{d})^{\otimes N}$ exists for any $N\in[8,11]$ by Theorem~\ref{thm:dpri}.
\end{enumerate}

By using Lemma~\ref{lem:rec}, we are able to list the existence of $4$-uniform states in $(\bbC^{d})^{\otimes N}$  in Table~\ref{table:4uni}. For the same discussion as above, we have  Table~\ref{table:5uni} for the existence of $5$-uniform states. The detail of arguments are provided in Appendix~\ref{appendix:5-uniform}. In the next section, we shall apply results in this section to the quantum information masking.

\section{$k$-uniform quantum information masking}\label{sec:k-uniformmasking}

In \cite{modi2018masking}, the authors proposed the concept of quantum information masking. An operation $\cS$ is said to mask quantum information contained in states $\{\ket{a_j}_{A_1}\in \cH_{A_1}\}$ by mapping them to states $\{\ket{\psi_j}\in \cH_{A_1}\otimes\cH_{A_2}\}$ such that  all the	reductions to one party of $\ket{\psi_j}$ are identical.
In \cite{li2018masking}, quantum information masking in multipartite systems is proposed. It is also required that all the reductions to one party of $\{\ket{\psi_j}\in \otimes_{\ell=1}^N\cH_{A_\ell}\}$ are identical. However, collusion between
some parties would then reveal the encoded quantum information. For example,  a masker $\cS$  masks quantum information contained in  $\{\ket{0},\ket{1},\ket{2}\}\in \bbC^3$ into $\{\ket{\psi_0},\ket{\psi_1},\ket{\psi_2}\}\in \bbC^3\otimes\bbC^3\otimes\bbC^3$, where $\ket{\psi_0}_{ABC}=\frac{1}{\sqrt{3}}(\ket{000}+\ket{111}+\ket{222})$, $\ket{\psi_1}_{ABC}=\frac{1}{\sqrt{3}}(\ket{021}+\ket{102}+\ket{210})$, and
$\ket{\psi_2}_{ABC}=\frac{1}{\sqrt{3}}(\ket{012}+\ket{120}+\ket{201})$. If Alice and Bob are collusive, then $\rho^{(0)}_{AB}=\tr_{C}\ketbra{\psi_0}{\psi_0}=\frac{1}{3}(\ketbra{00}{00}+\ketbra{11}{11}+\ketbra{22}{22})$, $\rho^{(1)}_{AB}=\tr_{C}\ketbra{\psi_1}{\psi_1}=\frac{1}{3}(\ketbra{02}{02}+\ketbra{10}{10}+\ketbra{21}{21})$, and $\rho^{(2)}_{AB}=\tr_{C}\ketbra{\psi_2}{\psi_2}=\frac{1}{3}(\ketbra{01}{01}+\ketbra{12}{12}+\ketbra{20}{20})$. Alice and Bob can easily distinguish $\rho^{(0)}_{AB}$, $\rho^{(1)}_{AB}$, and $\rho^{(2)}_{AB}$, and they would reveal the encoded quantum information.  To avoid this collusion, we propose the $k$-uniform quantum information masking as follows.

\begin{definition}\label{def:masking}
	An operation $\cS$ is said to $k$-uniformly mask quantum information contained in states $\{\ket{a_j}_{A_1}\in \cH_{A_1}\}$ by mapping them to states $\{\ket{\psi_j}\in\otimes_{\ell=1}^{N}\cH_{A_\ell}\}$ such that all the  reductions to $k$ parties of $\ket{\psi_j}$ are identical; i.e., for all $A=\{A_{\ell_1},A_{\ell_2},\ldots,A_{\ell_k}\}$ which is any subset of $\{A_1,A_2,\ldots,A_N\}$ with cardinality $k$,
	\begin{equation}
	\rho_{A}=\tr_{A^c}\ketbra{\psi_j}{\psi_j}
	\end{equation}
	has no information about the value of $j$. Specially, if $k=\fl{N}{2}$, we refer to the $k$-uniform quantum information masking as the strong quantum information masking.
\end{definition}

Our masking protocol also forms the basis for quantum secret sharing \cite{zukowski1998quest,cleve1999how}. The $k$-uniform quantum information masking allows secret sharing of quantum information from a ``boss'' to his ``subordinates'',  such that every collaboration between $k$ subordinates cannot retrieve the information. When $k=1$, the above definition is the same as \cite{modi2018masking,li2018masking}.
In fact, $\cS$ can be modeled by a unitary operator $U_{\cS}$ on the system $A_1$ plus some ancillary systems $\{A_2,A_3,\ldots,A_N\}$ and given by
\begin{equation*}
\cS: U_{\cS}\ket{a_j}_{A_1}\otimes\ket{b}_{A_1^c}=\ket{\psi_j}.
\end{equation*}
 The unitary operator  $U_{\cS}$ preserves orthogonality. See Fig.~\ref{Fig:defini9} for the sketch map of $k$-uniform quantum information masking. We assume that $\cH_{A_1}$ associates with the Hilbert space $\bbC^d$ and $U_{\cS}\ket{j}_{A_1}\otimes\ket{b}_{A_1^c}=\ket{\psi_j}$ for $0\leq j\leq d-1$. Then the masking process can be expressed by:
\begin{equation*}
\ket{j}\rightarrow \ket{\psi_j}, \quad \forall \ 0\leq j\leq d-1.
\end{equation*}
 Next, we give the relation between $k$-uniform states and $k$-uniform quantum information masking.

\begin{theorem}\label{thm:k-uniformmasking}
 If there exists a $(k+1)$-uniform states in $(\bbC^d)^{\otimes (N+1)}$, then all states in $\bbC^d$ can be $k$-uniformly masked in $(\bbC^d)^{\otimes N}$.
\end{theorem}
\begin{proof}
	We assume that $\ket{\psi}$ is a $(k+1)$-uniform state in $(\bbC^d)^{\otimes (N+1)}$. It means that any reduction to $k+1$ parties is maximally mixed. Let $\{\ket{j}\}_{i=0}^{d-1}$ be a computational basis in $\bbC^d$. We can write
	\begin{equation}
	\ket{\psi}=\frac{1}{\sqrt{d}}\sum_{j=0}^{d-1}\ket{j}\ket{\psi_j}.
	\end{equation}
	Since the reduction of the first party is also maximally mixed, i.e. $\rho_1=\frac{1}{d}I_d$, it implies $\braket{\psi_s}{\psi_t}=\d_{s,t}$ by Lemma~\ref{lem:orthonomal} in Appendix~\ref{appendix:orthonomal}. Then we can define a masker
	\begin{equation*}
	\cS: \ket{j}\rightarrow\ket{\psi_j}, \quad \forall \ 0\leq j\leq d-1.
	\end{equation*}
     The general superposition state $\sum_{j=0}^{d-1}\alpha_j\ket{j}\in \bbC^d$ should be mapped into $\sum_{j=0}^{d-1}\alpha_j\ket{\psi_j}\in (\bbC^d)^{\otimes N}$, where $\sum_{j=0}^{d-1} |\alpha_j|^2=1$. By  \cite[Proposition 2]{li2019k-uniform}, we obtain that any reduction to $k$ parties of $\sum_{j=0}^{d-1}\alpha_j\ket{\psi_j}$ is maximally mixed. Thus all states in $\bbC^d$ can be $k$-uniformly masked in $(\bbC^d)^{\otimes N}$.
\end{proof}
\vspace{0.4cm}

All the $k$-uniform states in Sec.~\ref{sec:k-uniform} and Sec.~\ref{sec:construction} can be used for $(k-1)$-uniform quantum information masking by Theorem~\ref{thm:k-uniformmasking}.

\begin{example}
Since
\begin{equation*}
\begin{aligned}
\ket{\psi}=\frac{1}{4}( &-\ket{000000}+\ket{001111}-\ket{010011}+\ket{011100}\\
&+\ket{000110}+\ket{001001}+\ket{010101}\\
&+\ket{011010}-\ket{111111}+\ket{110000}\\
&+\ket{101100}-\ket{100011}+\ket{111001}\\
&+\ket{110110}-\ket{101010}-\ket{100101})
\end{aligned}
\end{equation*}	
is a $3$-uniform state in $(\bbC^2)^{\otimes 6}$ \cite{goyeneche2014genuinely}, we can define a masker
\begin{equation*}
\begin{aligned}
\cS:\ket{0}\rightarrow \frac{1}{2\sqrt{2}}&(-\ket{00000}+\ket{01111}-\ket{10011}\\
&+\ket{11100}+\ket{00110}+\ket{01001}\\
&+\ket{10101})+\ket{11010});\\
\ket{1}\rightarrow
\frac{1}{2\sqrt{2}}&(-\ket{11111}+\ket{10000}+\ket{01100}\\
&-\ket{00011}+\ket{11001}+\ket{10110}\\
&-\ket{01010}-\ket{00101}).
\end{aligned}
\end{equation*}
By using the masker $\cS$,  all states in $\bbC^{2}$ can be $2$-uniformly (strongly) masked in $(\bbC^2)^{\otimes 5}$ by Theorem~\ref{thm:k-uniformmasking}.
\end{example}

 From Table~\ref{Table:k-unfiormresults}, we know that a $2$-uniform state  in $(\bbC^d)^{\otimes 4}$ exists for any $d\geq 3$ and $d\neq 6$,  then all states in $\bbC^d$ can be (1-uniformly) masked in $(\bbC^d)^{\otimes 3}$ for any $d\geq 3$ and $d\neq 6$. Thus  \cite[Corollary 2]{li2018masking} is a special case of Theorem~\ref{thm:k-uniformmasking}.  All states in $\bbC^d$ can be strongly masked in $(\bbC^d)^{\otimes 5}$ for any $d\geq 2$, since an $\ame$ state in $(\bbC^d)^{\otimes 6}$ exists for any $d\geq 2$ by Table~\ref{Table:k-unfiormresults}. Moreover, since an $\ame$ state in $(\bbC^d)^{\otimes 8}$ exists for any prime power $d\geq 5$ by Table~\ref{table:4uni}, all states in $\bbC^d$ can be strongly masked in $(\bbC^d)^{\otimes 7}$ for any prime power $d\geq 5$.  Thus, we have the following corollary.


\begin{figure*}[t]
	\begin{tikzpicture}
	\draw[rounded corners] (0,0) rectangle (4,2.7);
	\draw (2,2.4) node []{a $(k+1)$-uniform state  };
	\draw (2,1.9) node []{exists in $(\bbC^d)^{\otimes (N+1)}$};
	\draw[rounded corners] (0.8,0) rectangle (4,1.5);
	\draw (2.4,1) node []{an $\ame$ state exists};
	\draw (2.5,0.5) node []{ in $(\bbC^d)^{\otimes (N+1)}$};

	\draw[rounded corners] (6,0) rectangle (10.5,2.7);
	\draw (8.25,2.4) node []{all states in $\bbC^d$ can be };
	\draw (8.25,1.9) node []{ $k$-uniformly masked in $(\bbC^d)^{\otimes N}$  };
	
	\draw (7.7,1.3) node []{all states in $\bbC^d$ can be };
	\draw (7.7,0.8) node []{ strongly masked in };
	\draw (7.7,0.3) node []{  $(\bbC^d)^{\otimes N}$  };

	\draw[rounded corners] (12.5,0) rectangle (16.5,2.7);
	\draw (14.35,2.3) node []{a pure $((N,d,k+1))_d$  };
	\draw (14.35,1.8) node []{code};
	\draw[rounded corners] (12.5,0) rectangle (15.7,1.5);	
	\draw (14,1.18) node []{a pure  };
	\draw (14,0.68) node []{$((N,d,\fl{N}{2}+1))_d$};
	\draw (14,0.18) node []{code};
	
	\draw[rounded corners] (6,-1.5) rectangle (9.5,1.5);
	\draw (7.7,-0.2) node []{all states in $\bbC^{d_1}$ can be};
	\draw (7.7,-0.7) node []{ strongly masked in };
	\draw (7.7,-1.2) node []{ $\bbC^{d_1}\otimes\bbC^{d_2}\otimes\cdots\otimes\bbC^{d_N}$};	
	\draw [vecArrow](5.9,-0.75)--(3.7,-0.75);
	\draw (4.8,-1.1) node []{Proposition~\ref{pro:impossibleforeven} };
	\draw[rounded corners] (2,-1.3) rectangle (3.6,-0.2);
	\draw (2.75,-0.8) node []{$N$ is odd  };

	\draw [vecArrow](4.1,2)--(5.9,2);
	\draw (5,2.3) node []{Theorem~\ref{thm:k-uniformmasking} };
	
	\draw [vecArrow](4.1,0.6)--(5.9,0.6);
	\draw (5,0.9) node []{$N$ is odd };
	\draw (5,0.3) node []{Corollary~\ref{cor:AMEmasking} };

	\draw [vecArrow](12.4,2)--(10.6,2);
    \draw [vecArrow](12.4,0.6)--(9.6,0.6);	
	\draw (11.5,2.3) node []{Theorem~\ref{thm:QECC-and-masking} };
	\draw (11.5,0.9) node []{$N$ is odd };
    \draw (11.5,0.3) node []{Theorem~\ref{thm:QECC-and-masking} };
	\end{tikzpicture}
	\caption{The sketch map concludes the results in Sec.~\ref{sec:k-uniformmasking}.}\label{Fig:masking}
\end{figure*}
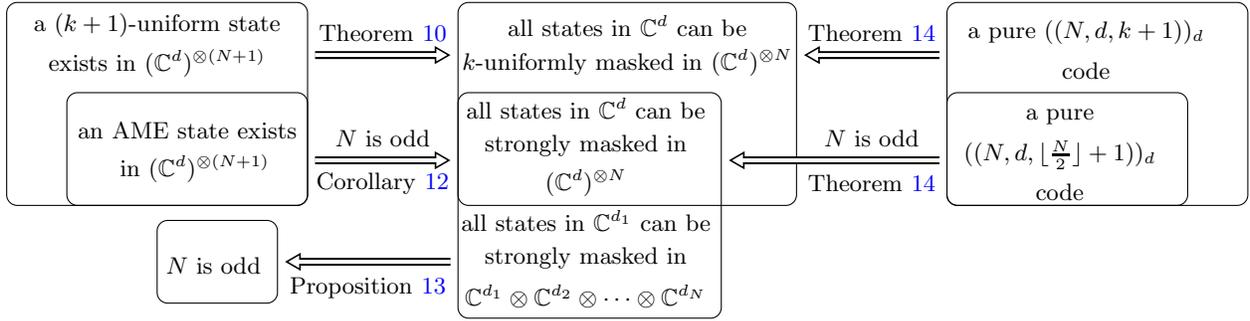

\begin{corollary}\label{cor:AMEmasking}
	When $N$ is odd, if there exists an $\ame$ state in $(\bbC^d)^{\otimes (N+1)}$, then all states in $\bbC^d$ can be strongly masked in $(\bbC^d)^{\otimes N}$.
\end{corollary}

From Theorem~\ref{thm:dpri}, we know that when $d\geq N$ is a prime power, there always exists an AME states in $(\bbC^d)^{\otimes N+1}$. These AME states can be used to strong quantum information masking when $N$ is odd by Corollary~\ref{cor:AMEmasking}.  However, when $N$ is even, strong quantum information masking is impossible.

\begin{proposition}\label{pro:impossibleforeven}
	When $N$ is even, an arbitrary  state in $\bbC^{d_1}$ cannot be strongly masked in $\bbC^{d_1}\otimes\bbC^{d_2}\otimes\cdots\otimes\bbC^{d_N}$.
\end{proposition}
\begin{proof}
	Assume that we can strongly mask two orthogonal states $\ket{s_0}$ and  $\ket{s_1}$ in $\bbC^{d_1}$ into $\bbC^{d_1}\otimes \bbC^{d_2}\otimes\cdots\otimes\bbC^{d_N}$, where $N$ is even. Let $\ket{s_0}\rightarrow \ket{\psi_0}$, $\ket{s_1}\rightarrow \ket{\psi_1}$, where $\ket{\psi_0},\ket{\psi_1}\in  \bbC^{d_1}\otimes \bbC^{d_2}\otimes\cdots\otimes\bbC^{d_N}$. Let  $\rho^{(0)}=\ketbra{\psi_0}{\psi_0}$, $\rho^{(1)}=\ketbra{\psi_1}{\psi_1}$, $A=\{A_1,A_2,\ldots, A_{\frac{N}{2}}\}$ and $B=\{A_{\frac{N}{2}+1},A_{\frac{N}{2}+2},\ldots,A_{N}\}$. Since  $\rho_A^{(0)}=\rho_A^{(1)}$, we can write $\ket{\psi_0}$ and $\ket{\psi_1}$ by Schmidt decomposition as
	\begin{equation*}\label{eq:schmidtdecom}
	\ket{\psi_0}=\sum_{j}\sqrt{\lambda_j}\ket{a_j}_A\ket{b_j}_B, \quad	\ket{\psi_1}=\sum_{j}\sqrt{\lambda_j}\ket{a_j}_A\ket{c_j}_B,
	\end{equation*}
	where $\lambda_j> 0$, $\{\ket{a_j}\}$, $\{\ket{b_j}\}$, $\{\ket{c_j}\}$ are all orthonormal sets.

	Assume  that we can mask the superposition state, $\ket{\psi}=u\ket{\psi_0}+v\ket{\psi_1}=\sum_{j}\sqrt{\lambda_j}\ket{a_j}(u\ket{b_j}+v\ket{c_j})$, where $|u|^2+|v|^2=1$.
	Since $\rho_{A}=\tr_B(\ketbra{\psi}{\psi})=\rho_{A}^{(0)}=\sum_{j}\lambda_j\ketbra{a_j}{a_j}$ by the masking condition, it implies that $(u^*\bra{b_s}+v^*\bra{c_s})(u\ket{b_t}+v\ket{c_t})=\delta_{s,t}$ by Lemma~\ref{lem:orthonomal} in Appendix~\ref{appendix:orthonomal}. Then we have
	\begin{equation}
	u^*v\braket{b_s}{c_t}+uv^*\braket{c_s}{b_t}=0, \quad \forall s,t.
	\end{equation}
	By setting $u=\frac{1}{\sqrt{2}},v=\frac{1}{\sqrt{2}}$ and $u=\frac{1}{\sqrt{2}},v=\frac{1}{\sqrt{2}}i$,  we obtain $\braket{b_s}{c_t}=0$ for all $s,t$. By the masking condition,
	\begin{equation}
	\rho_B^{(0)}=\sum_{j}\lambda_j\ketbra{b_j}{b_j}=\sum_{j}\lambda_j\ketbra{c_j}{c_j}=\rho_B^{(1)}.
	\end{equation}
	Then for any $j$, $\bra{b_j}	\rho_B^{(0)}\ket{b_j}=\bra{b_j}\rho_B^{(1)}\ket{b_j}$. It implies $\lambda_j=0$ for any $j$, which is impossible.
	Thus an arbitrary quantum state in $\bbC^{d_1}$ cannot be strongly masked in $\bbC^{d_1}\otimes\bbC^{d_2}\otimes\cdots\otimes\bbC^{d_N}$.
\end{proof}
\vspace{0.4cm}

Proposition~\ref{pro:impossibleforeven} can be related to no-go theorems \cite{wootters1982a,gisin1997optimal,lamaslinares2002experimental,barnum1996noncommuting,kalev2008no-broadcasting,pati2000impossibility,samal2011experimental,modi2018masking}, that is when $N$ is even, strong quantum information masking is impossible. From Proposition~\ref{pro:impossibleforeven}, we can also obtain that if all states in $\bbC^{d_1}$ can be strongly masked in $\bbC^{d_1}\otimes\bbC^{d_2}\otimes\cdots\otimes\bbC^{d_N}$, then $N$ is odd.
Actually, Theorem~\ref{thm:k-uniformmasking} can be generalized to heterogeneous systems, i.e. if there exists a $(k+1)$-uniform states in $\bbC^{d_1}\otimes\bbC^{d_1}\otimes\bbC^{d_2}\otimes\cdots\otimes\bbC^{d_N}$, then all states in $\bbC^{d_1}$ can be $k$-uniformly masked in $\bbC^{d_1}\otimes\bbC^{d_2}\otimes\cdots\otimes\bbC^{d_{N}}$. See \cite{goyeneche2016multipartite,shishenchenzhang} for some constructions of $k$-uniform states in heterogeneous systems.

Finally, we introduce the relation between qQECCs and  $k$-uniform quantum information masking.
	A \emph{$k$-uniform space} is a subspace of $(\bbC^d)^{\otimes N}$ in which every state is a $k$-uniform state.

\begin{theorem}\label{thm:QECC-and-masking}
	In the $k$-uniform quantum information masking scheme, we assume that the reduction states of $k$ parties are proportional to identity. The following statements are equivalent:
\begin{enumerate}[(i)]
	\item a pure $((N,d,k+1))_d$ QECC exists;
	\item all states in $\bbC^d$ can be $k$-uniformly masked in $(\bbC^d)^{\otimes N}$.
\end{enumerate}	
\end{theorem}
\begin{proof}
 ``$\Rightarrow$'' If there exits a pure $((N,d,k+1))_d$ QECC, then there exists a $k$-uniform space $\cU$ in $(\bbC^d)^{\otimes N}$ of dimension $d$ by Lemma~\ref{lem:code-and-k-uniformspace} in Appendix~\ref{appendix:orthonomal}. Assume that $\{\ket{\psi_j}\}_{j=0}^{d-1}$ is an orthonormal basis in $\cU$, and  $\{\ket{j}\}_{j=0}^{d-1}$ is a computational basis in $\bbC^d$. We can define a masker:
\begin{equation*}
\cS: \ket{j}\rightarrow\ket{\psi_j}, \quad \forall \ 0\leq j\leq d-1.
\end{equation*}
  The general superposition state $\sum_{j=0}^{d-1}\alpha_j\ket{j}\in \bbC^d$ should be mapped into $\sum_{j=0}^{d-1}\alpha_j\ket{\psi_j}\in (\bbC^d)^{\otimes N}$, where $\sum_{j=0}^{d-1} |\alpha_j|^2=1$. Since $\sum_{j=0}^{d-1}\alpha_j\ket{\psi_j}\in\cU$, it is a $k$-uniform state. Thus all states in $\bbC^d$ can be $k$-uniformly masked in $(\bbC^d)^{\otimes N}$.

 ``$\Leftarrow$''  We assume that all states in $\bbC^d$ can be $k$-uniformly masked in $(\bbC^d)^{\otimes N}$. Let $\cS$ be the masker, then the image of $\cS$ is a $k$-uniform space in $(\bbC^d)^{\otimes N}$ of dimension $d$ by definition. Thus there exists a pure $((N,d,k+1))_d$ QECC by Lemma~\ref{lem:code-and-k-uniformspace} in Appendix~\ref{appendix:orthonomal}.
\end{proof}
\vspace{0.4cm}

For a pure $((N,K,k+1))_d$ code, the quantum Singleton bound \cite{rains1999nonbinary} is
\begin{equation}\label{eq:singleton}
K\leq d^{N-2k}.
\end{equation}
When $N$ is even, $k=\frac{N}{2}$, and $K=d$, a $((N,d,\frac{N}{2}+1))_d$ code does not exist by Eq.~(\ref{eq:singleton}). When then reduction states of $k$ parties are proportional to identity, an arbitrary state in $\bbC^{d}$ cannot be strongly masked in $(\bbC^{d})^{\otimes N}$ by Theorem~\ref{thm:QECC-and-masking}, and it is a special case of Proposition~\ref{pro:impossibleforeven}. Finally, we should emphasize that Theorem~\ref{thm:k-uniformmasking} can be obtained by Theorem~\ref{thm:QECC-and-masking}. A $(k+1)$-uniform state in $(\bbC^d)^{\otimes(N+1)}$ is a pure $((N+1,1,k+2))_d$ code \cite{scott2004multipartite}. Since a pure $((N+1,1,k+2))_d$ code can imply a pure $((N,d,k+1))_d$ code \cite{rains1998quantum} (The converse is not true in generally), all states in $\bbC^d$ can be $k$-uniformly masked in $(\bbC^d)^{\otimes N}$ by Theorem~\ref{thm:QECC-and-masking}. Specially, when $N$ is odd,  an AME state in $(\bbC^d)^{\otimes (N+1)}$ is a pure $((N+1,1,\frac{N+1}{2}+1))_d$  code. It is equivalent to a pure $((N,d,\frac{N+1}{2}))_d$ code \cite{Huber2020quantumcodesof}. See Fig.~\ref{Fig:masking} for a summary of Sec.~\ref{sec:k-uniformmasking}.

\section{Conclusion}\label{sec:conclusion}
In this paper, we have given general constructions for $k$-uniform states by using linear codes, especially for $4,5$-uniform states. We have also given a new quantum information masking which is called $k$-uniform quantum information masking, and shown that $k$-uniform states and QECCs can be used for $k$-uniform quantum information masking. There are some interesting problems left. One open problem is to determine the existence of unknown $k$-uniform states in Tables~\ref{table:4uni}, \ref{table:5uni} and \ref{Table:k-unfiormresults}.  Besides, are there other methods that can $k$-uniformly mask all the
states of $\bbC^d$ in  quantum systems $(\bbC^d)^{\otimes N}$ for $2\leq k\leq \fl{N}{2}$?

\section*{Acknowledgments}
\label{sec:ack}	
We thank Felix Huber for providing us the idea of connecting quantum error-correcting codes to $k$-uniform quantum information masking. FS and XZ were supported by NSFC under Grant No. 11771419,  the Fundamental Research Funds for the Central Universities, and Anhui Initiative in Quantum Information Technologies under Grant No. AHY150200. M.-S.L. was supported by
NSFC (Grants No.
11875160,  No. 11871295 and No. 12005092). LC was supported by the  NNSF of China (Grant No. 11871089), and the Fundamental Research Funds for the Central Universities (Grant No. ZG216S2005).

\appendix

\section{Quantum error-correcting codes}\label{appendix:QECC}
 Let $\{e_{j}\}_{j=0}^{d^2-1}$  be an orthogonal operator basis for $\bbC^{d}$ that includes the identity $e_{0}=I$, such that $Tr(e_{i}^{\dag}e_{j})=\delta_{ij}d$.  On the $N$-partite quantum system $(\bbC^{d})^{\otimes N}$, a local error basis $\mathcal{E}$ consists of
\begin{equation*}
E_{\alpha}=e_{\alpha_{1}}\otimes e_{\alpha_{2}}\otimes \cdots \otimes e_{\alpha_{N}},
\end{equation*}
where $\alpha=(\alpha_{1},\alpha_{2},\cdots,\alpha_{N})\in \{0,1,\ldots,d^2-1\}^N$, each $e_{\alpha_{i}}$ acts on $\bbC^{d}$, and $Tr(E_{\alpha}^{\dag}E_{\beta})=\delta_{\alpha\beta}d^{n}$. The weight of a local error operator $E_{\alpha}$ is the number of $e_{i}$ which is not equal to identity, that is, $wt(E_{\alpha})=wt(\alpha)$.

Let $\mathcal{Q}$ be a $K$-dimensional subspace of $(\bbC^{d})^{\otimes N}$ spanned by the orthogonal basis $\{\ket{i}|i=0,1,\cdots,K-1\}$. Then $\mathcal{Q}$ is called an \emph{$((N,K,\delta))_{d}$ quantum error-correcting code} if for all $E\in\mathcal{E}$ with $wt(E)<\delta$,
	\begin{equation*}
	\langle i|E|j\rangle=\delta_{ij}C(E),
	\end{equation*}
	where the constant $C(E)$  depends only on $E$. Here $\delta$ is called the distance of the code. If $C(E)=d^{-N}Tr(E)$, then the code is called \emph{pure}. 
\vspace{0.4cm}

\section{The proof of the claim in  Lemma~\ref{lem;linirr}}\label{appendix:prooflinear}
\begin{proof}
	Let us introduce the \emph{generator matrix} and the parity check matrix for a linear code first. The generator matrix for an $[N,t,w]_d$ code $\cal{C}$ is any $t\times N$ matrix $G$ whose rows form a basis for $\cal{C}$. For a row vector $\bv\in \bbF_d^t$, a codeword  $\bc\in \cal{C}$ can be written as $\bc=\bv\cdot G$. The generator matrix has a standard form $G=[I_t|A]$, where $I_t$ is a $t\times t$ identity matrix and $A$ is a $t\times (N-t)$ matrix.  The \emph{parity check matrix} is an $(N-t)\times N$ matrix $H$ which satisfies $H\cdot \bc^{\mathrm{T}}=\mathbf{0}$ if and only if $\bc\in\cal{C}$. The parity check matrix for $\cal{C}$ can be written as $H=[-A^{\mathrm{T}}|I_{N-t}]$. If we consider the dual code $\cal{C}^{\bot}$, $H$ and $G$ are the generator and parity check matrices for $\cal{C}^{\bot}$, respectively. The linear code $\cal{C}$ has minimum distance $w$ if and only if every $w-1$ columns of  $H$ are linearly independent and some $w$ columns are linearly dependent \cite[Theorem 10, Chapter 1]{macwilliams1977theory}. Further, The dual code $\cal{C}^{\bot}$ has minimum distance $w^{\bot}$ if and only if every $w^{\bot}-1$ columns of  $G$ are linearly independent and some $w^{\bot}$ columns are linearly dependent. Now, we are ready to prove this claim.
	
	By \cite[Chapter 1.5.4]{huffman2010fundamentals}, we know that  $\mathcal{C}_1\oplus \mathcal{C}_2$ is an $[N_1+N_2,t_1+t_2,\min\{w_1,w_2\}]_d$ code. Assume $G_1$ and $G_2$ are generator matrices of $\mathcal{C}_1$ and $\mathcal{C}_2$, respectively. Then  $G=\begin{pmatrix}
	G_1 & 0\\
	0   & G_2
	\end{pmatrix}$ is the generator matrix of $\mathcal{C}_1\oplus \mathcal{C}_2$. Since the dual distance of $\mathcal{C}_j$  is $w_j^{\bot}$,  every $w_j^{\bot}-1$ columns of $G_j$ are linearly independent, and some $w_j^{\bot}$ columns are linearly dependent for each $j=1,2$. It implies that every  $\min\{w_1^{\bot}-1,w_2^{\bot}-1\}=\min\{w_1^{\bot},w_2^{\bot}\}-1$ columns of $G$  are linearly independent, and some $\min\{w_1^{\bot},w_2^{\bot}\}$ columns are linearly dependent. Thus the dual distance of $\mathcal{C}_1\oplus \mathcal{C}_2$ is $\min\{w_1^{\bot},w_2^{\bot}\}$.
\end{proof}
\vspace{0.4cm}

\section{The details for Table~\ref{table:5uni}}\label{appendix:5-uniform}
Some of our constructions are from  algebraic geometry codes (see \cite{algebraic2009} for definitions). For algebraic geometry codes, if there exist $N$ rational points and genus $g$  in Galois field $\bbF_q$, then there exists a linear code $[N,k,N-k+1-g]$ with dual distance $k+1-g$ for any  $g\leq k\leq N$ \cite{algebraic2009}. For example, there exist $18$ rational points and genus $2$  in $\bbF_8$ \cite{Manpoint}, then there exists a $[18,7,10]_8$ code with dual distance $6$, and hence exists a linear $\iro(r,N,8,5)$ for any $N\in [14,18]$ by Lemmas~\ref{dis} and \ref{lem:codearray}.
See  \cite{Manpoint} for  curves with many points. By Theorem~\ref{thm:dpri}, when $d\geq 18$ is a prime power, there exists a $5$-uniform state in $(\bbC^{d})^{\otimes N}$ for any $N\geq 10$. By \cite[Theorem 12]{feng2017multipartite},  there exists a $5$-uniform state in $(\bbC^d)^{\otimes N}$ for any prime $d\geq 2$ and $N\geq 18$. Then there exists a $5$-uniform state in $(\bbC^d)^{\otimes N}$ for any $d\geq 2$ and $N\geq 18$ by Lemma~\ref{lem:rec}.  Thus, we only need to consider $d<18$ that is a prime power, and $10\leq N\leq 17$.
\begin{enumerate}[{(i)}]
	\item When $d=2$, $5$-uniform states in $(\bbC^2)^{\otimes N}$ do not exist for each $N=10,11$  by Rains' bound \cite{rains1999quantum}. By \cite{AMEtable}, we know that there exists an $\oa(256,16,2,5)$. The minimum distance of the $\oa(256,16,2,5)$ is $6$ by using computer. Thus it is irredundant and there exists a $5$-uniform state  in $(\bbC^2)^{\otimes 16}$.
	
	\item When $d=3$,  a $5$-uniform state in $(\bbC^3)^{\otimes N}$ exists for any $N=10$ by \cite{AMEtable}, $N=12$ by a $[12,6,6]_3$ self-dual code in \cite{Selftable} and  $N\in[14,17]$ by \cite{feng2017multipartite},
	
	\item When $d=4$, a $5$-uniform state in $(\bbC^4)^{\otimes N}$ exists for any  $N=10$ by \cite{AMEtable}, and $N=12,14,16,17,18$ by  $[12,6,6]_4$, $[14,7,6]_4$, and $[18,9,8]_4$  self-dual codes in \cite{Selftable}.
	
	\item When $d=5$, a $5$-uniform state in $(\bbC^5)^{\otimes N}$ exists for any $N=10$ by \cite{AMEtable}, and $N\in[12,17]$ by \cite{feng2017multipartite}.
	
	\item When $d=7$,   a $5$-uniform state in $(\bbC^7)^{\otimes N}$ exists for any $N\in [10, 11]$ by \cite{AMEtable}, and $N\in[12,17]$ by \cite{feng2017multipartite}.
	
	\item When $d=8$, a $5$-uniform state in $(\bbC^8)^{\otimes N}$ exists for any $N\in [10,14]$ by \cite{AMEtable}, and $N\in[15,17]$ by an $[18,7,10]_8$ algebraic geometry code with dual distance $6$ in \cite{Manpoint}.
	
	\item When $d=9$, a $5$-uniform state in $(\bbC^9)^{\otimes N}$ exists for any $N=10,12$, $N\in [14,17]$ by Lemma~\ref{lem:rec}, $N=11$ by \cite{raissi2019constructing} and $N=13$ by a $[16,6,10]_9$ algebraic geometry code with dual distance $6$ in \cite{Manpoint}.
	
	\item When $d=11$,  a $5$-uniform state in $(\bbC^{11})^{\otimes N}$ exists for any $N\in [10,12]$ by Theorem~\ref{thm:dpri}, and $N\in[13,17]$ by a $[18,6,12]_{11}$ algebraic geometry code with dual distance $6$ in \cite{Manpoint}.
	
	\item When $d=13$,  a $5$-uniform state in $(\bbC^{13})^{\otimes N}$ exists for any $N\in [10,14]$ by Theorem~\ref{thm:dpri}, and $N\in[15,17]$ by a $[21,6,15]_{13}$ algebraic geometry code with dual distance $6$ in \cite{Manpoint}.
	
	\item When $d=16,17$,  a $5$-uniform state in $(\bbC^d)^{\otimes N}$ exists for any $N\in [10,17]$ by Theorem~\ref{thm:dpri}.
	
\end{enumerate}

By using Lemma~\ref{lem:rec}, we are able to list the existence of $5$-uniform states in $(\bbC^{d})^{\otimes N}$  in Table~\ref{table:5uni}.
\vspace{0.4cm}

\section{Two lemmas used in Sec.~\ref{sec:k-uniformmasking}}\label{appendix:orthonomal}
\begin{lemma}\label{lem:orthonomal}
	Assume $\ket{\psi}=\sum_{j}\sqrt{\lambda_j}\ket{a_j}\ket{b_j}$, $\lambda_j>0$ and $\{\ket{a_j}\}$ is an orthonomal set,   then  $\rho_A=\sum_{j}\lambda_j\ketbra{a_j}{a_j}$ if and only if $\braket{b_s}{b_t}=\delta_{s,t}$.
\end{lemma}
\begin{proof}
	``$\Leftarrow$'' Obviously.
	
	``$\Rightarrow$'' Since
	\begin{equation}
	\rho_{A}=\tr_{B}\ketbra{\psi}{\psi}=\sum_{j,\ell}\sqrt{\lambda_j\lambda_\ell}\braket{b_\ell}{b_j}(\ketbra{a_j}{a_\ell}),
	\end{equation}
	we have
	\begin{equation}
	\begin{aligned}
	\bra{a_t}\rho_A\ket{a_s}&=\bra{a_t}(\sum_{j,\ell}\sqrt{\lambda_j\lambda_\ell}\braket{b_\ell}{b_j}(\ketbra{a_j}{a_\ell}))\ket{a_s}\\
	&=\sqrt{\lambda_t\lambda_s}\braket{b_s}{b_t}.
	\end{aligned}
	\end{equation}
	Moreover,
	\begin{equation}
	\bra{a_t}\rho_A\ket{a_s}=\bra{a_s}(\sum_{j}\lambda_j\ketbra{a_j}{a_j})\ket{a_t}=\lambda_s\d_{s,t}.
	\end{equation}
	It implies $\lambda_s\d_{s,t}=\sqrt{\lambda_t\lambda_s}\braket{b_s}{b_t}.$ Hence $\braket{b_s}{b_t}=\delta_{s,t}$.
\end{proof}
\vspace{0.4cm}

\begin{lemma}\label{lem:code-and-k-uniformspace}\cite{Huber2020quantumcodesof}
	The following objects are equivalent:
	\begin{enumerate}[(i)]
		\item a pure $((N,K,k+1))_d$ QECC;
		\item a $k$-uniform space in $(\bbC^d)^{\otimes N}$ of dimension $K$.
	\end{enumerate}
\end{lemma}
\vspace{0.4cm}
\bibliographystyle{IEEEtran}
\bibliography{reference}
\end{document}